\documentclass{sig-alternate}
\usepackage{amsmath,amsfonts, graphicx}
\usepackage{tikz}
\usepackage{slashbox}

\usepackage{todonotes} %% Remove for final version!

\DeclareMathOperator{\diag}{diag}
\DeclareMathOperator{\colspace}{colspace}
\newcommand{\ordinal}{\raisebox{0.7ex}{\scriptsize th}}

\begin{document}
\newcommand{\C}{\ensuremath{\mathbb C}}
\newcommand{\D}{\ensuremath{\mathbb D}}
\newcommand{\R}{\ensuremath{\mathbb R}}
\newcommand{\Z}{\ensuremath{\mathbb Z}}
\newcommand{\N}{\ensuremath{\mathbb{N}}}
\newcommand{\Q}{\ensuremath{\mathbb{Q}}}

\newcommand{\Ma}{\textsc{Maple}}
\newcommand{\Mat}[3]{{#1}^{#2\times#3}}
\newcommand{\ID}{\mathbf{1}}

\newtheorem{theorem}{Theorem}
\newtheorem{remark}{Remark}
\newtheorem{algorithm}{Algorithm}

\setcopyright{acmcopyright}
\doi{10.475/123_4}
\isbn{123-4567-24-567/08/06}

%Conference
\conferenceinfo{ISSAC'16}{July 19--22, 2016, Waterloo, ON, Canada}

\acmPrice{\$15.00}

%
% --- Author Metadata here ---
\CopyrightYear{2016}
% --- End of Author Metadata ---

\title{Matrix factoring by fraction-free reduction}
%\titlenote{(Produces the permission block, and

\numberofauthors{2}

\author{
\alignauthor
Johannes Middeke\\
       \affaddr{RISC Linz}\\
       \affaddr{Johannes Kepler University}\\
       \affaddr{Linz, Austria}\\
       \email{jmiddeke@risc.uni-linz.ac.at}
% 2nd. author
\alignauthor
David J. Jeffrey\\
       \affaddr{ORCCA and Dept Applied Mathematics}\\
       \affaddr{The University of Western Ontario}\\
       \affaddr{London, Ontario, Canada N6A 5B7}\\
       \email{djeffrey@uwo.ca}
}

\date{\today}

\maketitle
\begin{abstract}
We consider exact matrix decomposition by Gauss-Bareiss reduction. We investigate two aspects of the process: common row and column factors and the influence of pivoting strategies. We identify two types of common factors: systematic and statistical. Systematic factors depend on the process, while statistical factors depend on the specific data. We show that existing fraction-free QR (Gram--Schmidt) algorithms create a common factor in the last column of Q. We relate the existence of row factors in LU decomposition to factors appearing in the Smith normal form of the matrix. For statistical factors, we identify mechanisms and give estimates of the frequency. Our conclusions are tested by experimental data. For pivoting strategies, we compare the sizes of output factors obtained by different strategies. We also comment on timing differences.
\end{abstract}


\section{Introduction}
Although known earlier, fraction-free methods for exact matrix computations
became popular after Bareiss's study of Gaussian elimination \cite{Bareiss:1968}.
Extensions to related topics, such as LU factoring, were considered in
\cite{LeeSaunders:1995, NakosTurnerWilliams:1997, ZhouJeffrey2008}.
Gram--Schmidt orthogonalization and QR factoring were studied by \cite{EKM:1996}, under the more descriptive name of exact division.
Recent studies have looked at extending fraction-free LU factoring to non-invertible matrices~\cite{Jeffrey2010LU} and rank profiling \cite{DumasPernetSultan2015}, and more generally to areas such as the Euclidean algorithm, and the Berlekamp--Massey algorithm~\cite{Kaltofen:2013:BM}.
We consider matrices over an integral domain $\D$.
For the purposes of giving illustrative examples and conducting computational
experiments, matrices over $\Z$ and $\Q[x]$ are used,
because the metrics associated with these domains are well established and
familiar to readers.
We emphasize, however, that the methods here 
apply for all integral domains, as opposed to methods that target
specific domains,
such as~\cite{GiesbrechtStorjohann2002,PauderisStorjohann2013}.

The starting point for this paper is the fraction-free form for LU
decomposition~\cite{Jeffrey2010LU}: given a matrix $A$ over an integral domain $\D$,
\begin{equation}\label{def:LDU}
A = P_w LD^{-1}U P_c\ ,
\end{equation}
where $L$, $D$ and $U$ are over $\D$. 
$L$ and $U$ are lower and upper triangular and their diagonals contain the pivots
of the Gaussian elimination; $D$ is diagonal and contains products of the pivots.  
The permutation matrices $P_w$ and $P_c$
ensure that the decomposition is always a full-rank decomposition, even if $A$
is rectangular or rank deficient.
In addition to the usual indeterminacy due to varying pivot choices, the columns of $L$ and the rows of $U$ can be multiplied by common factors,
which then appear also in $D$. We show in section \ref{sec:QR} that this form can cover $QR$ decomposition also.

Our first main result is for QR factoring.
In this context, the orthonormal $Q$ matrix used in floating point calculations
is replaced by a $\Theta$ matrix, which is left-orthogonal, i.e. $\Theta^t\Theta$ is
diagonal, but $\Theta\Theta^t$ is not. We show that for a square matrix $A$, the last column of $\Theta$, as calculated by existing algorithms, is subject to an exact division by the determinant of $A$, with a significant reduction in size.
This is an example of a systematic factor, being one inherent to the algorithm.

Systematic factors occur in several ways. The Bareiss algorithm
uses exact division precisely to remove systematic factors; the Gram--Schmidt
algorithm from \cite{EKM:1996} is another, where exact division removes
systematic factors during the reduction.
In addition to these, we add a
different type of systematic factor: we show a relation between GCDs existing for the rows in matrices obtained from LU factoring, and entries in the Smith normal form of the same initial matrix.

We next consider statistical factors: ones which depend on the initial data.
When $LU$ and $QR$ matrices are computed using current standard
fraction-free algorithms, their rows and columns may contain common factors.
We discuss their origins and show we can predict a significant proportion of them from simple considerations.
Their presence influences aspects such as uniqueness.
%This is because fraction-free methods aim at avoiding
%computations in the quotient field of the original domain,
%in spite of this being ultimately
%impossible, since the solutions being sought can only be expressed as quotients.
%Nonetheless, fraction-free methods delay, for as long as possible,
%the ultimate fall from grace of the solution method.
Specifically, for the basic decomposition \eqref{def:LDU},
we show how common factors can be moved between the three matrices.
We discuss when this is beneficial. 

We next consider the role of pivoting in Gaussian reduction.
Geddes \textsl{et al.}~\cite{GeddesCzaporLabahn:1992} comment
``We also mention that when the entries of $A^{(k)}$ are not of uniform size, it may be worthwhile to interchange rows in order to obtain a smaller pivot at the next step".
It is often said that whereas for \textit{floating-point} Gaussian elimination the
largest pivot should be chosen, in the setting of \text{exact computation} the smallest pivot is best. Although within the floating-point literature, pivoting has been studied over an extended period,
much less attention has been paid to the question in the context of exact computation.
We consider a number of strategies empirically, and show that selecting the smallest pivot, suitably defined, leads to smaller output matrices.

The paper will start with a brief discussion of fraction-free methods, then present results for QR factoring, LU factoring, and finally pivoting.

\section{Fraction-free methods}
Fraction-free methods are based on the assumption that
it is more efficient to compute with the elements of the input domain
of a matrix than to compute with elements from the quotient field.
Since the solutions usually sought require the quotient field, 
fraction-free methods can be regarded as delaying for as long as possible 
the ultimate fall from grace of the computation.
Here, some measurements are reported to supply empirical evidence
to support fraction-free methods.

Our first point of comparison is between the LU decomposition offered by
\Ma, through \texttt{LUDecomposition(A)}, 
and our own implementation of the $L
D^{-1} U$ decomposition based on \cite{Jeffrey2010LU}. The built-in
\Ma\ command returns matrices $L$ and $U$ such that all diagonal
elements of $L$ are $1$, and both $L$ and $U$ contain elements from
the quotient field of $\D$. The procedure which we implemented has the
output format described in \cite[Theorem~2]{Jeffrey2010LU}.

Figure~\ref{fig:LUvsGauss.size.int} shows the ratio of average storage
requirements for the decomposition of integer matrices. Here, we
measure the total number of digits needed to represent the final
output. Note that this metric does not depend on the internal
implementation of the two functions, nor does it depend on the
particular computer algebra system. As
figure~\ref{fig:LUvsGauss.size.int} illustrates, fraction-free methods
require roughly half the storage.

Table~\ref{tab:LUvsGauss.time.int} compares timings for random integer
matrices, as functions of matrix size and initial data size. For this
experiment we used our own implementation of Gaussian elimination,
since we do not know the details of \Ma's built-in procedure,
which may well use compiled code. By writing our own programme we
make sure that every common part, for example pivot searching, uses exactly
the same code and only the reduction steps differ. As
table~\ref{tab:LUvsGauss.time.int} reveals, the advantages of
the fraction-free method are clear, while not spectacular.

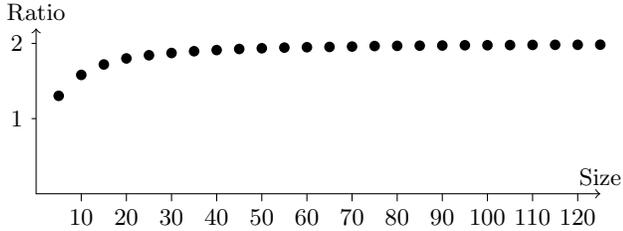
\begin{figure}
  \centering
  %   n  | LU         | Gauss      | Gauss/LU
  % -----+------------+------------+------------
  %    5 | 146.508000 | 190.908000 |   1.303055
  %   10 | 896.240000 | 1417.36400 |   1.581456
  %   15 | 2779.12400 | 4781.05600 |   1.720346
  %   20 | 6355.86800 | 11445.9760 |   1.800852
  %   25 | 12215.1920 | 22512.2560 |   1.842972
  %   30 | 20961.1000 | 39256.0800 |   1.872806
  %   35 | 33210.7360 | 62982.5840 |   1.896453
  %   40 | 49563.4640 | 94772.5640 |   1.912146
  %   45 | 70663.3920 | 136122.020 |   1.926344
  %   50 | 97178.9600 | 188140.572 |   1.936022
  %   55 | 129770.940 | 252351.724 |   1.944593
  %   60 | 169093.772 | 329685.856 |   1.949722
  %   65 | 215715.264 | 421467.828 |   1.953816
  %   70 | 270329.120 | 529504.964 |   1.958742
  %   75 | 333529.260 | 655598.340 |   1.965640
  %   80 | 406152.264 | 799330.596 |   1.968056
  %   85 | 488953.140 | 963353.712 |   1.970237
  %   90 | 582221.728 | 1148419.34 |   1.972478
  %   95 | 686935.360 | 1357178.93 |   1.975701
  %  100 | 804151.692 | 1589237.40 |   1.976291
  %  105 | 933573.856 | 1846974.41 |   1.978391
  %  110 | 1076595.96 | 2132222.00 |   1.980522
  %  115 | 1233838.07 | 2445327.97 |   1.981887
  %  120 | 1405828.04 | 2786695.06 |   1.982245
  %  125 | 1593172.79 | 3160086.70 |   1.983518
\begin{tikzpicture}[x=0.06cm, step=1]
  \draw[<->] (0,2.2) node[above] {Ratio} |- (125,0) node[above] {Size};
  \foreach \y in {1,2} \draw (0,\y) -- (-1,\y) node[left] {$\y$};
  \foreach \x in {10,20,...,120} \draw (\x,0) -- (\x,-0.1) node[below] {$\x$};
  \fill (5,1.303055) circle (2pt);
  \fill (10,1.581456) circle (2pt);
  \fill (15,1.720346) circle (2pt);
  \fill (20,1.800852) circle (2pt);
  \fill (25,1.842972) circle (2pt);
  \fill (30,1.872806) circle (2pt);
  \fill (35,1.896453) circle (2pt);
  \fill (40,1.912146) circle (2pt);
  \fill (45,1.926344) circle (2pt);
  \fill (50,1.936022) circle (2pt);
  \fill (55,1.944593) circle (2pt);
  \fill (60,1.949722) circle (2pt);
  \fill (65,1.953816) circle (2pt);
  \fill (70,1.958742) circle (2pt);
  \fill (75,1.965640) circle (2pt);
  \fill (80,1.968056) circle (2pt);
  \fill (85,1.970237) circle (2pt);
  \fill (90,1.972478) circle (2pt);
  \fill (95,1.975701) circle (2pt);
  \fill (100,1.976291) circle (2pt);
  \fill (105,1.978391) circle (2pt);
  \fill (110,1.980522) circle (2pt);
  \fill (115,1.981887) circle (2pt);
  \fill (120,1.982245) circle (2pt);
  \fill (125,1.983518) circle (2pt);
\end{tikzpicture}
\caption{Comparison of the output size of Gaussian Elimination vs.\ $L
  D^{-1} U$. %
  \normalfont%
  We show the ratio of the number of digits in the output of Gaussian
  elimination divided by the number of digits in the output of the $L
  D^{-1} U$ decomposition for random square integer matrices of
  various sizes.}
  \label{fig:LUvsGauss.size.int}
\end{figure}

\begin{table}
  \centering
  \small
  \begin{tabular}{r|*{11}{r}}
    \backslashbox{$s$}{$n$}
    & 11    & 19    & 31    & 53    & 73    & 97    & 107   \\\hline
3   & 1.00 & 0.85 & 0.91& 0.78 & 0.72 & 0.63 & 0.62 \\
7   & 0.97 & 0.88 & 0.83& 0.71 & 0.65 & 0.59 & 0.56 \\
13  & 0.94 & 0.84 & 0.82& 0.68 & 0.61 & 0.55 & 0.52 \\
23  & 0.93 & 1.33 & 0.85& 0.66 & 0.59 & 0.51 & 0.49 \\
37  & 0.89 & 0.81 & 0.77& 0.63 & 0.56 & 0.49 & 0.47 \\
53  & 0.93 & 0.80 & 0.74& 0.62 & 0.55 & 0.47 & 0.45 \\
67  & 0.90 & 1.32 & 0.73& 0.60 & 0.54 & 0.46 & 0.44 \\
89  & 0.89 & 0.53 & 0.72& 0.61 & 0.53 & 0.45 & 0.43 \\
109 & 0.87 & 0.77 & 0.73& 0.60 & 0.52 & 0.44 & 0.42 \\
  \end{tabular}
  \normalsize
  \caption{Timings of $L D^{-1} U$ decomposition vs.\ Gaussian elimination
    for integer matrices. \normalfont The table shows the run times for
    $L D^{-1} U$
    divided by those for Gauss for random $n$-by-$n$ matrices with maximal
    entry size $s$. \label{tab:LUvsGauss.time.int}}
\end{table}

\section{Common factors in QR}
\label{sec:QR}
A fraction-free (exact division) algorithm for Gram--Schmidt orthogonalization
was described by \cite{EKM:1996}. An algorithm based on $LU$ factoring has been described
in \cite{PursellTrimble:1991,ZhouJeffrey2008}. The two approaches yield the same results.
We denote the decomposition by $A=\Theta D^{-1} R$, because $Q$ usually
denotes an orthonormal matrix, and $\Theta$ is not orthonormal.
We give a new statement of the basic theorem.

\begin{theorem}\label{thm:L=U^t}
Given a square, full-rank matrix $A$ over an integral domain $\D$, the partitioned matrix $(A^tA,A^t)$ has a fraction-free LU decomposition
\[ (A^tA, A^t ) = R^t D^{-1}(R, \Theta^t) \ , \]
where $\Theta^t\Theta=D$ and $A=\Theta D^{-1} R$.
\end{theorem}
\begin{proof}
We can apply $L U$ factoring, to get
\[(A^tA,A^t) = \hat L D^{-1} (\hat U, \Theta)\ ,\]
where the notation $\hat L,\hat U$ emphasizes that the matrices refer not to a factoring of
$A$, but of $A^tA$.
Since this matrix is symmetric we obtain
\begin{displaymath}
  \hat L D^{-1} \hat U = A^t A = \hat U^t D^{-1} \hat L^t.
\end{displaymath}
Because $A$ has full rank, so do $L$ and $U$ and we can rewrite the
equation as
\begin{displaymath}
  \hat U (\hat L^t)^{-1} D = D \hat L^{-1} \hat U^t.
\end{displaymath}
Examination of the matrices on the left hand side reveals that they
and therefor also their product are all upper triangular. Similarly,
the left hand side is a lower triangular matrix and the equality of
the two implies that they must both be diagonal. Cancelling
$D$ and rearranging the equation yields
\begin{math}
  \hat U = (\hat L^{-1} \hat U^t) \hat L^t
\end{math}
where $\hat L^{-1} \hat U^t$ is diagonal. This shows that the rows of
$\hat U$ are just multiples of the rows of $\hat L^t$. However, we
know that the diagonal entries of $\hat U$ and $\hat L$ are the
same. Thus, $\hat L^{-1} \hat U^t$ is the identity and $\hat L = \hat
U^t$.

We now write $R = \hat L^t = \hat U$. The proof of
\cite[Theorem~8]{ZhouJeffrey2008} shows $A = \Theta D^{-1} R$ and
\begin{math}
  \Theta^t \Theta = \hat U (D \hat L^{-1})^t.
\end{math}
Expanding the last expression and using the definition of $R$ gives
then
\begin{math}
  \Theta^t \Theta = R R^{-1} D = D.
\end{math}
\end{proof}

We now give an explicit expression of the last column of $\Theta$, showing the
common factor of $\det A$.

\begin{theorem}\label{thm:detA}
  With $A \in \D^{n\times n}$ and $\Theta$ as in theorem~\ref{thm:L=U^t},
  we have for all $i=1,\ldots,n$ that
  \begin{displaymath}
    \Theta_{in} = (-1)^{n+i} \det A_{in} \det A
  \end{displaymath}
  where $\det A_{in}$ is the $(i,n)$ minor of $A$.
\end{theorem}
\begin{proof}
  We use the notation from the proof of theorem~\ref{thm:L=U^t}.
  From $\Theta D^{-1} \hat L^t = A$ we obtain
  \begin{displaymath}
    \Theta^t A
    = D \hat L^{-1} A^t A
    = D \hat L^{-1} (\hat L D^{-1} \hat U)
    = \hat U.
  \end{displaymath}
  Thus, since $A$ has full rank, $\Theta^t = \hat U A^{-1}$ or,
  equivalently,
  \begin{displaymath}
    \Theta
    = (\hat U A^{-1})^t
    = (A^{-1})^t \hat U^t
    = (\det A)^{-1} (\operatorname{adj} A)^t \hat U^t
  \end{displaymath}
  where $\operatorname{adj} A$ is the adjugate matrix of $A$. Since
  $\hat U^t$ is a lower triangular matrix with $\det A^t A = (\det
  A)^2$ at position $(n,n)$, the claim follows.
\end{proof}

\begin{theorem}\label{thm:QR.cancel}
Given a square matrix $A$, a reduced fraction-free $QR$ decomposition is
given by $A=\hat\Theta \hat D^{-1}\hat R$, where
$S=\operatorname{diag}(1,1,\ldots,\det A)$ and
$\hat \Theta = \Theta S^{-1}$, and $\hat R=S^{-1}R$.
In addition, $\hat D=S^{-1}DS^{-1}=\hat\Theta^t \hat\Theta$.
\end{theorem}

\begin{proof}
By theorem \ref{thm:detA}, $\Theta S^{-1}$ is an exact division. The theorem follows from
$A=\Theta S^{-1} S D^{-1} S S^{-1} R$.
\end{proof}

As an example we consider the $4$-by-$4$ integer matrix
\begin{displaymath}
  A = \begin{pmatrix}
    -62 & 21  & 64  & -96 \\
    38  & 18  & 31  & 56  \\
    -59 & -86 & 19  & 2   \\
    40  & -91 & -62 & 9
  \end{pmatrix}.
\end{displaymath}
Computing the $Q R$ decomposition with theorem~\ref{thm:L=U^t} yields
\begin{displaymath}
  \Theta =
  \begin{pmatrix}
    -62 & 268341  & 2658137038 & -23374523883001 \\
    38  & 155634  & 8243861790 & 3112061098992   \\
    -59 & -843590 & 2460946816 & 14218033256642  \\
    40  & -976219 & -81659738  & -18215371009147
  \end{pmatrix}
\end{displaymath}
$D = \operatorname{diag}(10369, 1760876458298, 81089877269400184044,$\\
\phantom{D=diag(>> }$ 1090005501728694354954965838)$\ , \\
$R=$
\begin{displaymath}
  \begin{pmatrix}
    10369 & 816       & -6391        & 8322         \\
    0     & 169821242 & 66495846     & -27518383    \\
    0     & 0         & 477501379182 & 210662060582 \\
    0     & 0         & 0            & 2282727441742609
  \end{pmatrix}.
\end{displaymath}
We can now check that indeed $\det A = 47777897$ divides the last
column of $\Theta$.

Cancelling $\det A$ from the last column of $\Theta$ and the last
entry of $R$ as well as reducing $D$ accordingly leads to the much
simpler output
\begin{displaymath}
  \Theta =
  \begin{pmatrix}
    -62 & 268341  & 2658137038 & -489233 \\
    38  & 155634  & 8243861790 & 65136   \\
    -59 & -843590 & 2460946816 & 297586  \\
    40  & -976219 & -81659738  & -381251
  \end{pmatrix},
\end{displaymath}
$D = \operatorname{diag}(10369, 1760876458298, 81089877269400184044,$\\
\phantom{D=diag(>> }$477501379182)$ and
\begin{displaymath}
  R =
  \begin{pmatrix}
    10369 & 816       & -6391        & 8322         \\
    0     & 169821242 & 66495846     & -27518383    \\
    0     & 0         & 477501379182 & 210662060582 \\
    0     & 0         & 0            & 47777897
  \end{pmatrix}
\end{displaymath}

\section{Common factors in LU}
Given a matrix $A$ over an integral domain $\D$, we consider the fraction-free
decomposition $A=LD^{-1}U$.
It is clear that if the elements in a column of $L$ or a row of $U$
possess a common GCD, then that factor can be removed,
reducing the size of the matrix elements.
We identify 3 sources of common GCDs.

\subsection{Input data}
The initial matrix may contain one or more rows having a common GCD,
usually because of modelling choices made by the user.
Standard Gaussian elimination will then transfer the common factor into all
subsequent rows. If several rows have different GCDs, then all GCDs accumulate
in subsequent rows.

\subsection{ LU and the Smith Form}

The following theorem links the Smith normal form of a given matrix
with factors appearing in the LU decomposition.
\begin{theorem}\label{thm:smith}
  Let $A \in \Mat{\D}nn$ have the Smith normal form $S =
  \diag(d_1,\ldots,d_n)$ where $d_1,\ldots,d_n \in \D$. Moreover, let
  $A = L D^{-1} U$ be an $L D^{-1} U$ decomposition of $A$. Then for
  $k=1,\ldots,n$
  \begin{equation*}
    d_k^* = \prod_{j=1}^k d_j \mid U_{k,*}
    \quad\text{and}\quad
    d_k^* \mid L_{*,k}.
  \end{equation*}
\end{theorem}
\begin{remark}
  The values $d_1^*, \ldots, d_n^*$ are known as the
  \emph{determinantal divisors} of $A$.
\end{remark}
\begin{proof}
  According to \cite[II.15]{Newman:1972}, the diagonal entries of the
  Smith form are quotients of the determinantal divisors, i.\,e.,
  $d_1^* = d_1$ and $d_k = d^*_k/d^*_{k-1}$ for
  $k=2,\ldots,n$. Moreover, $d_k^*$ is the greatest common divisor of
  all $k$-by-$k$ minors of $A$ for each $k=1,\ldots,n$. Thus, we only
  have to prove that the entries of the $k$\ordinal\ row of $U$ are
  $k$-by-$k$ minors of $A$. However, this follows from
  \cite[Eqns~(9.8), (9.12)]{GeddesCzaporLabahn:1992}, since the
  $k$\ordinal\ row of $U$ are just
  \begin{displaymath}
    \det
    \begin{pmatrix}
      A_{11} & \cdots & A_{1k} & A_{1j} \\
      \vdots &  & \vdots & \vdots \\
      A_{k1} & \cdots & A_{kk} & A_{kj} \\
    \end{pmatrix}
    \quad\text{where}\quad j=1,\ldots,k.
  \end{displaymath}

  Similarly, following the algorithm in \cite{Jeffrey2010LU}, we see
  that the columns of $L$ are just made up by copying entries from the
  columns of $U$ during the reduction. More precisely, the
  $k$\ordinal\ column of $L$ will have the entries $a_{1k}^{(k-1)},
  \ldots, a_{nk}^{(k-1)}$ (using the notation of
  \cite{GeddesCzaporLabahn:1992}). But these are again just $k$-by-$k$
  minors of $A$.
\end{proof}

We give an example using the domain $\Q[x]$.
Let $A$ be the polynomial matrix
\begin{displaymath}
  \begin{pmatrix}
     -\frac32 & -x^3 + 5 x^2 + 3 x - \frac92 & x^2 + x     & \frac12 x^3 - x^2  \\[1ex]
     -3       & -2 x^3 + 10 x^2 + 5 x - 9   & 2 x^2 + 2 x & x^3 - 2 x^2         \\[1ex]
     \frac12  & x^3 + \frac32               & 0           & -\frac12 x^3        \\[1ex]
     -\frac12 & -x - \frac32                & 0           & \frac12 x
  \end{pmatrix}.
\end{displaymath}
The Smith normal form $S$ of $A$ is
\begin{displaymath}
  \diag(1, x, x(x+1), x(x+1)(x-1))
\end{displaymath}
and thus its determinantal divisors are $d_1^* = 1$, $d_2^* = x$,
$d_3^* = x^2(x+1)$ and $d_4^* = x^3(x+1)^2(x-1)$. Computing the $L
D^{-1} U$ decomposition of $A$ yields $A = L D^{-1} U$ where $L$ is
\begin{displaymath}
  \scriptsize
  \begin{pmatrix}
    -\frac32 & 0 & 0 & 0 \\[1ex]
    -3 & \frac32 x & 0 & 0 \\[1ex]
    \frac12 & -x^3  - \frac52 x^2  - \frac32 x & \frac12 x^3  + \frac12 x^2  & 0 \\[1ex]
    -\frac12 & -\frac12 x^3  + \frac52 x^2  + 3 x & -\frac12 x^3  - \frac12 x^2 &
    -\frac14 x^6  - \frac14 x^5  + \frac14 x^4  + \frac14 x^3
  \end{pmatrix},
\end{displaymath}
$D = \diag(-3/2, -9/4 x, 3/4 x^4 + 3/4 x^3, -1/8 x^9 - 1/4 x^8 + 1/4
x^6 + 1/8 x^5)$, $U$ is
\begin{displaymath}
  \scriptsize
  \begin{pmatrix}
    -\frac32 & -x^3  + 5 x^2  + 3 x - \frac92 & x^2  + x & \frac12 x^3  - x^2 \\[1ex]
    0 & \frac32 x & 0 & 0 \\[1ex]
    0 & 0 & \frac12 x^3 + \frac12 x^2 & -\frac12 x^4 - \frac12 x^3 \\[1ex]
    0 & 0 & 0 & -\frac14 x^6 - \frac14 x^5 + \frac14 x^4 + \frac14 x^3
  \end{pmatrix}.
\end{displaymath}
Computing the column factors of $L$ and the row factors of $U$ yields
$1$, $x$, $x^2(x+1)$ and $x^3(x-1)(x+1)^2$, i.\,e., exactly the
determinantal divisors. In general, there could be other factors as
well.

\subsection{Statistical effects}

Suppose that during Bareiss's algorithm after $k-1$ iterations we have
reached the following state
\begin{displaymath}
  A^{(k-1)} =
  \begin{pmatrix}
    U & \underline{*} & \underline{*} & \mathbf{*} \\
    \overline{0} & p & * & \overline{*} \\
     0 & 0 & a & \overline{v} \\
     0 & 0 & b & \overline{w} \\
     \mathbf{0} & \underline{0} & \underline{*} & \mathbf{*}
  \end{pmatrix}\ ,
\end{displaymath}
where $U$ is an upper triangular matrix, $p,a,b \in \D$,
$\overline{v}, \overline{w} \in \D^{1\times n-k-1}$ and the other
overlined quantities are row vectors and the underlined quantities are
column vectors. Assume that $a \neq 0$ and that we choose it as a
pivot. Continuing the computations we now eliminate $b$ (and the
entries below) by cross-multiplication
\begin{displaymath}
  A^{(k-1)} \leadsto
  \begin{pmatrix}
    U & \underline{*} & \underline{*} & \mathbf{*} \\
    \overline{0} & p & * & \overline{*} \\
     0 & 0 & a & \overline{v} \\
     0 & 0 & 0 & a\overline{w} - b\overline{v} \\
     \mathbf{0} & \underline{0} & \underline{0} & \mathbf{*}
  \end{pmatrix}.
\end{displaymath}
Here, we can see that any common factor of $a$ and $b$ will be a
factor of every entry in that row, i.\,e., $\gcd(a,b) \mid
a\overline{w} - b\overline{v}$. However, we still have to carry out
the exact division step. This leads to
\begin{displaymath}
  A^{(k-1)} \leadsto
  \begin{pmatrix}
    U & \underline{*} & \underline{*} & \mathbf{*} \\
    \overline{0} & p & * & \overline{*} \\
     0 & 0 & a & \overline{v} \\
     0 & 0 & 0 & \frac1p(a\overline{w} - b\overline{v}) \\
     \mathbf{0} & \underline{0} & \underline{0} & \mathbf{*}
  \end{pmatrix} = A^{(k)}.
\end{displaymath}
The division by $p$ is exact. Some of the factors in $p$ might be
factors of $a$ or $b$ while others are hidden in $\overline{v}$ or
$\overline{w}$. However, every common factor of $a$ and $b$ which is
not also a factor of $p$ will still be a common factor of the
resulting row. In other words,
\begin{displaymath}
  \Bigl.
  \frac{\gcd(a,b)}{\gcd(a,b,p)}
  \;\Bigm|\;
  \frac1p(a\overline{w} - b\overline{v})
  \Bigr..
\end{displaymath}

In fact, the factors do not need to be tracked during the $L D^{-1} U$
reduction but can be computed afterwards: All the necessary entries
$a$, $b$ and $p$ of $A^{(k-1)}$ will end up as entries of $L$. More
precisely, we will have $p = L_{k-1,k-1}$, $a = L_{k,k}$ and $b =
L_{k+1,k}$.

If $\D$ are the integers, then the probability that the quotient
$\gcd(a,b)/\gcd(p,a,b) \neq 1$, i.e. nontrivial, for random $a,b,p$ equals 
$1 - 6\zeta(3)/\pi^2 \approx 26.92\%$ \cite{Hare, Winterhof}. 
Thus, for
integer matrices these factors occur with a high enough frequency to
suggest we care about them. 
In our experiments we saw that independently of the
size of the input matrix this method could detect about $40.17\%$ of
all the common prime row factors occurring in $U$.\footnote{This
  experiment was carried out with random square matrices $A$ of sizes
  between $5$-by-$5$ and $125$-by-$125$. We decomposed $A$ into $P_w L
  D^{-1} U P_c$ and then computed the number of predicted prime
  factors in $U$ and related that to the number of actual prime
  factors. We did not consider the last row of $U$ since this contains
  only the determinant.}

As an example consider the matrix
\begin{displaymath}
  A = \begin{pmatrix}
    0   & -18 & -92 & -25 & -60 \\
    49  & -77 & 66  & 45  & 8   \\
    18  & 31  & 69  & -81 & 51  \\
    -58 & 41  & 22  & 37  & -97 \\
    -77 & -52 & 48  & -19 & -10
  \end{pmatrix}.
\end{displaymath}
This matrix has a $L D^{-1} U$ decomposition with
\begin{displaymath}
  L =
  \begin{pmatrix}
    8   &     0   &       0    &        0    &          0   \\
    -10 &    -126 &         0  &          0  &            0\\
    51  &  -2355  &   134076   &         0   &           0\\
    -97 &    4289 &   -233176  &  -28490930  &            0\\
    -60 &    2940 &   -148890  &  -53377713  &  11988124645
  \end{pmatrix}
\end{displaymath}
and
\begin{displaymath}
  U =
  \begin{pmatrix}
    8  &    49 &       45 &         -77 &            66 \\
    0  &  -126 &      298 &       -1186 &          1044 \\
    0  &     0 &   134076 &     -414885 &        351648 \\
    0  &     0 &        0 &   -28490930 &      55072620 \\
    0  &     0 &        0 &           0 &   11988124645
  \end{pmatrix}.
\end{displaymath}
The method outlined above correctly predicts the common factor $2$ in
the second row, the factor $3$ in the third row and the factor $2$ in
the fourth row. However, it does not detect the additional factor $5$
in the fourth row.

\medskip

There is another way in which common factors in integer matrices can
arise: Let $d$ be any number. Then for random $a,b$ the probability
that $d \mid a+b$ is $1/d$. That means that if $v,w \in \Z^{1\times
  n}$ are vectors, then $d \mid v + w$ with a probability of
$1/d^n$. This effect is noticable in particular for small numbers like
$d = 2,3$ and in the last iterations of the $L D^{-1} U$ decomposition
when the number of non-zero entries in the rows has shrunk. For
instance, in the second last iterations we only have three rows with
at most three non-zero entries each. Moreover, we know that the first
non-zero entries of the rows cancel during cross-multiplication. Thus,
a factor of $2$ appears with a probability of $25\%$ in one of those
rows, a factor of $3$ with a probability of $11.11\%$. In the example
above, the probability for the factor $5$ to appear in the fourth row
was $4\%$.

\medskip

In a manner similar to theorem~\ref{thm:QR.cancel}, we can cancel all factors which we
find from the final output:

\begin{theorem}\label{thm:LU.cancel}
  Given a matrix $A \in \Mat{\D}mn$ with rank $r$ and its
  decomposition $A = P_w L D^{-1} U P_c$, if $D_U =
  \diag(d_1,\ldots,d_r)$ is a diagonal matrix with $d_k \mid
  \gcd(U_{k,*})$, then setting $\hat U = D_U^{-1} U$ and $\hat D = D
  D_U^{-1}$ where both matrices are fraction-free we have the
  decomposition $A = P_w L \hat D^{-1} \hat U P_c$.
\end{theorem}
\begin{proof}
  By \cite[Theorem~2]{Jeffrey2010LU} the diagonal entries of $U$ are the
  pivots chosen during the decomposition and they also divide the
  diagonal entries of $D$. Thus, any common divisor of $U_{k,*}$ will
  also divide $D_{kk}$ and therefor both $\hat U$ and $\hat D$ are
  fraction-free. We can easily check that $A = P_w L D^{-1} D_U
  D_U^{-1} U = P_w L \hat D^{-1} \hat U P_c$.
\end{proof}

\begin{remark}
  If we can find common column factors of $L$ we can cancel them in
  the same way. However, if we have already cancelled factors from $U$, 
  then there is no guarantee that $d \mid L_{*,k}$ implies $d
  \mid \hat D_{kk}$. Thus, in general we can only cancel $\gcd(d,
  \hat D_{kk})$ from $L_{*,k}$.
\end{remark}

\section{Pivoting strategies for LU}
Our pivoting strategies are all based on full pivoting, which is
already implied by the definition of the form.  We define a number of
pivoting strategies.
\begin{description}
\item[Largest] We select the largest pivot according to an appropriate
  metric.  Metrics were the absolute value for integer matrices and
  the degree as well as the height for matrices univariate
  polynomials.
\item[Smallest] Here we select the smallest pivot according
  to the same metrics as above.
\item[First] We select the first non-zero pivot.
\item[Factors] With this strategy we select the pivot which has the
  least number of prime factors counted with multiplicity.
\end{description}
Of course, the ``factors'' strategy is not viable in practice since
the factorisation is much too costly. However, it does provide
interesting theoretical insight.

In contrast to floating point calculations, accuracy of the result is
not an issue, and we consider instead the size of the elements in the
matrices generated, and any impact on the efficiency of the
computation. By size we examine the following
\begin{description}
\item[Digits] For integer matrices or matrices we count the total
  number of base-$10$ digits needed to represent it. We also use this
  measurement for matrices with rational number entries where we
  simply add up the digits of the numerators and the denominators.
\item[Terms] For univariate polynomial matrices we count the total
  number of non-zero terms in the fully expanded representation of the
  entries.
\item[Height] As another metric for polynomial matrices we use the
  maximal height of its entries.
\item[Factors] For both integer and polynomial matrices we measure the
  total number of row factors. Here, we compute the greatest common
  divisor of each row and count the number of prime factors with
  multiplicity. The number of factors for ech row is then added up.
\end{description}
Note that the measured quantities do solely depend on the output. In
particular do they not depend on how the programme handles its memory
during the computations. Also note that the measurements are chosen in
such a way that they are independent of the internal representation of
the data. For instance, every programme has to store all the digits of
the output matrices somehow.

%\todo[inline,caption={}]{To Do:
%  \begin{itemize}
%    % \item Describe nomenclature for pivoting strategies.
%  \item Describe effect: for floating point they affect accuracy, but
%    for exact computation they affect size of output and efficiency.
%    % \item Describe experiments performed.
%  \item Summarize results.
%  \end{itemize}}
%
%\medskip

The experiments included in this paper were all carried out with
\Ma. We use our own implementation of the $L D^{-1} U$ decomposition
which closely follows \cite{Jeffrey2010LU}. For each experiment we
generated random matrices $A$ of different sizes and then performed
the decomposition $A = P_w L D^{-1} U P_c$ using the strategies
described above. That is, each random matrix $A$ was decomposed with
each of the strategies. We then applied the applicable
measurements. In the end we computed the mean value of all the
results. More precise description of the experiments follow below.

For table~\ref{tab:pivoting.int} we generated three hundred integer
matrices for each size. The entries where between $-11^3$ and
$11^3$. Also in order to be closer to real world problems, we made
sure that the sizes of the entries in our matrices varied widely with
less than $25\%$ of the entries reaching maximal
size. Table~\ref{tab:pivoting.int} shows the number of digits and the
number of row factors of $U$ where the decompositions are done using
the ``smallest'', ``largest'' and ``factors'' strategies described at
the beginning of this section.

\begin{table}
  \centering
  \small
  \begin{tabular}{r|*{3}{r}|*{3}{r}}
    $n$ & \multicolumn{3}{c}{digits} & \multicolumn{3}{c}{row factors} \\
    & smallest & largest & factors
    & smallest & largest & factors
    \\\hline
    5  & 78.13   & 101.74  & 85.13   & 7.58  & 8.01  & 5.74 \\
    10 & 503.72  & 678.40  & 569.40  & 11.65 & 12.80 & 6.44 \\
    15 & 1625.08 & 2130.83 & 1833.94 & 17.17 & 17.95 & 7.77 \\
    20 & 3832.33 & 4888.83 & 4297.05 & 21.38 & 22.88 & 7.98 \\
    25 & 7533.28 & 9365.39 & 8316.27 & 26.06 & 27.92 & 8.26
  \end{tabular}
  \normalsize
  \caption{Output sizes for different pivoting strategies for integer matrices.
    \normalfont
    The table compares the average number of digits and the number of
    row common factors of $U$ for random $n$-by-$n$ integer matrices $A$
    as input using the ``smallest'', ``largest'' and ``factors'' pivoting
    strategies.}
  \label{tab:pivoting.int}
\end{table}

Table~\ref{tab:pivoting.polynom} shows a similar experiment for
matrices of univariate polynomials. We compare the strategies of
choosing the pivot with the smallest degree versus choosing the
largest degree and choosing the smallest height. The matrices $A$
contained random polynomials with integer coefficients between $-100$
and $100$ and degree at most $3$. During the same experiment we also
measured the number of row factors and the height of $U$ but we did
not find a significant difference between the different strategies.

\begin{table}
  \centering
  \begin{tabular}{r|*{3}{r}}
    $n$ & smallest degree & largest degree & height  \\\hline
    5   & 83.07           & 106.80         & 91.91   \\
    10  & 532.45          & 698.15         & 609.13  \\
    15  & 1696.09         & 2154.53        & 1946.16 \\
    20  & 3932.09         & 4860.95        & 4504.71
  \end{tabular}
  \caption{Output sizes for different pivoting strategies for polynomial
    matrices. \normalfont
    The table compares the number of terms of $U$ for random $n$-by-$n$
    input matrices $A$ using the  ``smallest degree'', ``largest degree''
    and ``smallest height'' pivoting strategies.}
  \label{tab:pivoting.polynom}
\end{table}

\section{Solving}

%% Solving by applying two LU decompositions.

In this section we detail a method for solving linear systems in such
a way that fractions are delayed until the final output.

Let $A \in \Mat{\D}mn$ and $b \in \D^m$. We wish to solve the system
$A x = b$, seeking solutions $x$ with entries in the field of
fractions of $\D$. First, apply the $L D^{-1} U$ decomposition as in
\cite{Jeffrey2010LU}. We obtain
\begin{displaymath}
  D L^{-1} P_w^t \, A
  = \begin{pmatrix}
    V \\
    W
  \end{pmatrix}
  A
  = \begin{pmatrix}
    U & B \\
    0 & 0 \\
  \end{pmatrix}
  P_c
  \quad\text{and}{\quad}
  P_c x =
  \begin{pmatrix}
    y \\
    z
  \end{pmatrix},
\end{displaymath}
where all (sub) matrices have entries in $\D$, $U$ is an $r$-by-$r$,
regular and upper triangular matrix, $r$ is the rank of $A$ and where
$y$ has dimension $r$.
Then $A x = b$ if and only if $W b = 0$ and $U
y + B z = V b$.

Now, perform a second $L D^{-1} U$ decomposition on $U$
(pivoting is not needed as all diagonal entries of $U$ are non-zero),
working from the bottom to the top, and from right to
left\footnote{More formally, let $\Pi$ be the matrix of the
  permutation which maps $i$ to $r+1-i$ and decompose $\Pi U \Pi$ in
  the normal way applying the same permutations to the result.}. This
will compute a regular $X \in \Mat{\D}rr$ such that $X U = \Delta$ is
a diagonal matrix. Then $A x = b$ if and only if $W b = 0$ and $\Delta
y + X B z = X V b$.

Assume now that the compatibility condition $W b = 0$ is fulfilled. In
order to compute a particular solution $x_0$ of the system $A x = b$,
we can simply choose
\begin{displaymath}
  x_0 = P_c^{-1}
  \begin{pmatrix}
    \Delta^{-1} X V b \\
    0
  \end{pmatrix}
  =
  \tilde\Delta^{-1} S b
  \quad\text{where}\quad
  S = P_c
  \begin{pmatrix}
    X V\\
    0
  \end{pmatrix}
\end{displaymath}
and where $\tilde\Delta = P_c \diag(\Delta,\ID) P_c$ is a diagonal
matrix with entries in $\D$.

Moreover, we can compute the nullspace of $A$ in the following way: If
\begin{displaymath}
  x \in \colspace
  P_c^{-1} \begin{pmatrix}
    -\Delta^{-1} \, X B \\
    \ID_{n-r}
  \end{pmatrix}
  ,
\end{displaymath}
then we can easily check that $A x = 0$. Since the $n-r$ columns of
the matrix spanning the space are clearly linearly independent, it
follows that this is already the entire nullspace of $A$. Thus,
setting
\begin{displaymath}
  K = P_c
  \begin{pmatrix}
    -X B \\ \ID
   \end{pmatrix},
\end{displaymath}
we see the nullspace of $A$ is $\colspace \tilde\Delta^{-1} K$, with
$\tilde\Delta$ as defined above.

Note that $S$ and $K$ are both matrices over $\D$. Thus, the
particular solution and the nullspace are both computed in a
fraction-free way. Moreover, neither of the matrices depends on the
right hand side $b$. Consequently, after computing $W$, $S$,
$\tilde\Delta$ and $K$, we can solve the system $A x = b$ for
arbitrary $b$ by just checking whether $W b = 0$ and then computing
$x_0 = \tilde\Delta^{-1} S b$.

\medskip

We summarise the method as follows:

\begin{algorithm}\label{alg:solve}
  \begin{description}
  \item[Input:] A matrix $A \in \Mat{\D}mn$.
  \item[Output:] Matrices $W$, $S$, and $K$ with entries in $\D$ and a
    diagonal matrix $\tilde\Delta$ with entries in $\D$ such that for
    any $b \in \D^m$ if the compatibility condition $W b = 0$ is met,
    then the system $A x = b$ has the solution set $\tilde\Delta^{-1}
    S b + \colspace \tilde\Delta^{-1} K$.
  \item[Steps:]
    \begin{enumerate}
    \item Apply the $L D^{-1} U$ decomposition to obtain
      \begin{displaymath}
        D L^{-1} P_w^t A
        =
        \begin{pmatrix}
          V \\ W
        \end{pmatrix}
        A
        =
        \begin{pmatrix}
          U & B \\
          0 & 0
        \end{pmatrix}
      \end{displaymath}
      where $U$ is upper triangular.
    \item Use a backwards $L D^{-1} U$ decomposition on $U$ to obtain
      a matrix $X$ such that diagonal $X U = \Delta$ is a diagonal
      matrix.
    \item Let
      \begin{displaymath}
        S = P_c
        \begin{pmatrix}
          X V \\ 0
        \end{pmatrix},
        \quad
        K = P_c
        \begin{pmatrix}
          -X B \\ \ID
        \end{pmatrix}
      \end{displaymath}
      and $\tilde\Delta = P_c \diag(\Delta,\ID) P_c$.
    \end{enumerate}
  \end{description}
\end{algorithm}

\medskip

As an example we consider the matrix
\begin{displaymath}
  A = \begin{pmatrix}
    -370  & -62  & -101 & -3 \\
    -708  & -120 & -193 & -5 \\
    -304  & -50  & -83  & -3 \\
    -1962 & -336 & -534 & -12
  \end{pmatrix}
\end{displaymath}
and examine the two systems below for solutions.
\begin{displaymath}
  A x =
  \begin{pmatrix}
    1 \\ 0 \\ 0 \\ 1
  \end{pmatrix}
  = b_1
  \quad\text{and}\quad
  A x =
  \begin{pmatrix}
    0 \\ 0 \\ 1 \\ 1
  \end{pmatrix}
  = b_2.
\end{displaymath}

Following algorithm~\ref{alg:solve}, we first compute
\begin{displaymath}
  \begin{pmatrix}
      1   &   0   &   0   & 0       \\
      3   &   0   &  -3   & 0     \\
      110 &   -36 &   -50 &   0 \\\hline
      7   &  -6   &   -1  &  1
  \end{pmatrix}
  A =
  \left(\begin{array}{ccc|c}
      -3 &   -62  &  -101 &   -370\\
      0  &  -36   &  -54  &  -198\\
      0  &    0   &  -12  &   -12 \\\hline
      0  &     0  &    0  &     0
  \end{array}\right) P_c
\end{displaymath}
where $P_c$ represents the permutation $(1\quad 4)$; and use this to
define the matrices $V$, $W$, $U$ and $B$. Next, we compute
\begin{displaymath}
  X =
  \begin{pmatrix}
    432 &   -744 &   -288 \\
    0   &  -12   &   54 \\
    0   &    0   &    1
  \end{pmatrix}
\end{displaymath}
and $X U = \diag(-1296, 432, -12) = \Delta$. This leads to
\begin{displaymath}
  S =
  \begin{pmatrix}
    0     &   0     &   0      & 0 \\
    5904  &  -1944  &  -2664   &  0 \\
    110   &   -36   &   -50    &  0 \\
    -33480&    10368&    16632 &   0
  \end{pmatrix},
  \quad
  K =
  \begin{pmatrix}
    1 \\
    -1728 \\
    12 \\
    9072
  \end{pmatrix}
\end{displaymath}
and $\tilde\Delta = \diag(1,432,-12,-1296)$.

We can check that $W b_1 = 8 \neq 0$. Consequently, the system $A x =
b_1$ does not have a solution. On the other hand, $W b_2 = 0$ and the
solution set for $A x = b_2$ is
\begin{displaymath}
  \tilde\Delta^{-1} S b + \colspace \tilde\Delta^{-1} K
  =
  \begin{pmatrix}
    0 \\ -37/6 \\ 25/6 \\ -77/6
  \end{pmatrix}
  +
  \colspace
  \begin{pmatrix}
    1 \\ -4 \\ -1 \\ -7
  \end{pmatrix}.
\end{displaymath}

%%%\input{MiddekeJeffrey.solving}

%%% Local Variables:
%%% mode: latex
%%% TeX-master: "MiddekeJeffrey"
%%% End:

\section{Conclusions}
We have shown that fraction-free LU and QR decompositions can contain 
significant common factors, and we have shown how these can be beneficially
removed to obtain more compact decompositions.
Moreover, their removal makes the decomposition unique.

We considered removing the common factors as soon as they can be detected during the
computation of the decompositions.  
This would require either discovering the GCDs by direct computation, or by predicting them by different, preferably simpler, computations.
Although we have displayed here mechanisms that generate common factors, and which lend themselves to predictions through relatively simple calculations, there are other mechanisms which we have not discussed. These require more extensive computations to predict, and quickly leave the realm of
reasonable strategies. Therefore we have concluded that it is most sensible to leave
common factor identification to the final stage of decomposition.

We hope that reduced decompositions can be implemented as the standard form in 
future computer-algebra systems.

%ACKNOWLEDGMENTS are optional
\section{Acknowledgments}

This work was supported in part by the Austrian Science Fund (FWF)
grant SFB50 (F5009-N15).

We would like to thank Prof.\ Kevin G.\ Hare and Univ.-Doz.\ Dr.\ Arne
Winterhof for helpful.

% \bibliographystyle{abbrv}
% \bibliography{MiddekeJeffrey}

\end{document}